\documentclass[12pt]{amsart}

\usepackage[textheight=23.5cm, left=1.2cm, right=1.2cm]{geometry}

\usepackage[english]{babel}
\usepackage[T1]{fontenc}
\usepackage[utf8]{inputenc}

\theoremstyle{plain}
    \newtheorem{theorem}{Theorem}[section]
    \newtheorem{lemma}[theorem]{Lemma}
    \newtheorem{proposition}[theorem]{Proposition}
    \newtheorem{corollary}[theorem]{Corollary}

\theoremstyle{definition}
    \newtheorem{remark}[theorem]{Remark}

    \newtheorem{example}[theorem]{Example}

\usepackage[
    draft = false,
    unicode = true,
    colorlinks = true,
    allcolors = blue,
    hyperfootnotes = true,
    citecolor = red
]{hyperref}
\usepackage{amsfonts,amsmath,amssymb,amscd,amsthm,url}
\usepackage[inline]{enumitem}
\usepackage{graphicx,epsf,afterpage}
\usepackage{soul}
\usepackage{multicol}
\usepackage{array}
\usepackage{braket}
\usepackage{epigraph}
\usepackage{rotating}

\usepackage[
backend=biber,
style=numeric-comp,
sorting=none,
giveninits=true
]{biblatex}
\usepackage{csquotes}

\usepackage{xcolor}

\usepackage{ulem}

\usepackage{tikz}
\usetikzlibrary{positioning, calc, intersections, through, arrows, matrix, chains, math}
\usetikzlibrary{shadows}
\usetikzlibrary{graphs}
\usetikzlibrary{graphs.standard}
\usetikzlibrary{patterns}
\usetikzlibrary{decorations.pathreplacing,calligraphy,backgrounds}

\newcommand\norm[1]{\ensuremath{\left\lVert#1\right\rVert}}
\newcommand\abs[1]{\ensuremath{\left\lvert#1\right\rvert}}

\renewcommand{\Pr}{\mathrm{P}}

\DeclareMathOperator{\Expect}{\mathbb{E}}

\DeclareMathOperator{\Ker}{Ker}

\newcommand{\Lcal}{\mathcal{L}}
\newcommand{\Hcal}{\mathcal{H}}
\newcommand{\Tcal}{\mathcal{T}}
\newcommand{\Fcal}{\mathcal{F}}

\newcommand{\defing}[1]{\textbf{\emph{#1}}}

\newcommand{\R}{\ensuremath{\mathbb{R}}}

\newcommand{\Cx}{\ensuremath{\mathbb{C}}}

\newcommand{\N}{\ensuremath{\mathbb{N}}}

\renewcommand{\geq}{\geqslant}

\renewcommand{\leq}{\leqslant}

\newcounter{mcnt}

\newcounter{wordcnt}

\begin{document}

\title{Banach and counting measures, \\ and dynamics of singular quantum states \\ generated by averaging of operator random walks}

\author{E.\,A. Dzhenzher, S.\,V. Dzhenzher, V.\,Zh. Sakbaev}

\begin{abstract}
     In this paper the random channels and their compositions in the space of quantum states are studied.
     For compositions of i.i.d. random unitary channels, the limit behaviour of probability distributions is described.
     The sufficient condition for convergence in probability is obtained.
     The generalized convergence in distribution w.r.t. weak operator topology is obtained.
     The analysis of transmission of pure and normal states to the set of singular states is done. The dynamics of quantum states is described in terms of the evolution of the values of quadratic forms of operators from the algebra that implements the representation of canonical commutation relations.
\end{abstract}

\thanks{\hspace{-4mm}E.\,A. Dzhenzher: catbordacheva@mail.ru
\\
S.\,V. Dzhenzher: sdjenjer@yandex.ru. orcid: 0009-0008-3513-4312
\\
V.\,Zh. Sakbaev: fumi2003@mail.ru. orcid: 0000-0001-8349-1738
\\
V.\,Zh. Sakbaev: Keldysh Institute of Applied Mathematics of Russian Academy of Sciences 125047, Miusskaya pl. 4, Moscow, Russia
\\
All authors: Moscow Institute of Physics and Technology 141701, Institutskii per. 9, Dolgoprudny, Russia}

\maketitle
\thispagestyle{empty}


\emph{Keywords: Banach measure; counting measure; shift semigroup; quantum dynamics; singular quantum states.}

\vspace{5mm}

\emph{MSC: 46G12, 47A56, 47D03, 81P16.}

\newcommand{\one}{{\ensuremath{\mathbf{1}}}}
\newcommand{\Mbf}{{\ensuremath{\mathbf{M}}}}
\newcommand{\Sbf}{{\ensuremath{\mathbf{S}}}}
\newcommand{\Hbf}{{\ensuremath{\mathbf{H}}}}
\newcommand{\Ubf}{{\ensuremath{\mathbf{U}}}}
\newcommand{\Vbf}{{\ensuremath{\mathbf{V}}}}
\newcommand{\Fbf}{{\ensuremath{\mathbf{F}}}}
\newcommand{\Tbf}{{\ensuremath{\mathbf{T}}}}
\newcommand{\Phbf}{{\ensuremath{\mathbf{\Phi}}}}

\section{Introduction}

The field of random quantum channels and random quantum dynamical semigroups attracts significant interest from mathematical physics.
Random unitary channels and random quantum dynamical semigroups naturally arise in the consideration of the dynamics of open quantum systems \cite{Accardi-Volovich-Lu-2002, Kempe, Teretenkov2023}. In particular, they arise in the ambiguity of the quantization procedure \cite{OSS-2016}. For example, in \cite{AmosovSakbaev}, random semigroups of shifts are considered from the point of view of averaging over some measure.
The limit theorems for the compositions of i.i.d. random operators have applications in differential equations \cite{Kalmetev, Zamana, Sakbaev-Smolyanov-Zamana, Loboda2019} and in the description of the dynamics of open quantum systems \cite{Aharonov-D-Z, Pechen, DzhenzherSakbaev25-SLLN, Dzhenzher25, AmosovSakbaev}.
The limit behavior of the sequence of quantum random walks of i.i.d. unitary channels is described in \cite{Joye, GOSS-2022, DzhenzherSakbaev26-evol}.
Another law of large numbers for quantum dynamics was considered in \cite{Kolokoltsov2021, KolokoltsovTroeva2022}.

The unitary representation of the canonical commutation relations in the Hilbert space of functions square-integrable w.r.t. the Lebesgue measure is one of the fundamental points in the construction of quantum mechanics \cite{Holevo2011-ch2}. The dynamics of quantum states under the influence of quantum random walks was investigated in the works \cite{Holevo-qprobstat, Aharonov-D-Z, Pechen}. Quantum kinetic equations of the GKSL type, describing the quantum dynamical semigroup in the space of quantum states, were studied in the works \cite{Holevo-gens-diff}. Of interest is the question of the relationship between the well-posedness of the initial-boundary value problem describing the dynamics of a quantum state and the preservation of the normality property of the state. In the works \cite{Sakbaev2016-Cauchy, Volovich-Sakbaev-2018, Glazatov-Orlov-Sakbaev-2025}, it was shown that the boundedness of the time interval of existence of a solution to the Schrödinger equation indicates the singularization of the quantum state upon crossing the boundary of the solution's existence interval.

To obtain a unitary representation of random walks on a real line, the real axis is equipped with a counting measure.
The counting measure on a real line is a non-negative shift-invariant countably additive measure, which takes finite values only on finite subsets.
In the Hilbert space of square integrable w.r.t. counting measure functions, we consider the groups of unitary operators.
In particular, the unitary group of argument shifts and Weyl unitary representation of the Canonical Commutation Relations are considered.
The analogues of such groups for the classical Lebesgue measure on the real line were considered in \cite{Orlov-Sakbaev-25, DzhenzherSakbaev26-evol}.
It was shown there that the random unitary group actions on the pure states converge to other random pure states in different probabilistic modes.
Here, we show that a unitary group of shifts is not strongly continuous. The mean value of unitary representation of a one-parameter semigroup of Gaussian measures on the real line is the discontinuous semigroup of self-adjoint contractions. Also, the representation of Gaussian random walks by the random unitary channel in the space of quantum states is considered. 

We show that the mean value of the one-parameter semigroup of Gaussian measures on the real axis in the space of quantum states is the semigroup of singular quantum channels that map any quantum state to a singular state.
We describe the evolution of mean values of the semigroup of singular quantum channels on operators of a special algebra of bounded linear operators in the considered Hilbert space.
This algebra of operators includes operators of multiplications of a bounded measurable functions and operators of convolutions with a countable additive measures with a bounded variation. In this algebra of operators, canonical commutation relations are represented.

We note that in \cite{Volovich-Sakbaev-2018}, using regularization of an ill-posed initial-boundary value problem for the Schr\"{o}dinger equation with a maximal symmetric but not self-adjoint operator, the limit dynamics of states on the algebra of observables generated by the Abelian algebra of operators of multiplication by a bounded measurable function and a ring of compact operators was obtained. The resulting limit dynamics of a state also described the continuous transition of a quantum state to a set of singular states. The continuous dynamics of singular quantum states obtained in the present work is given by averaging unitary quantum channels and describes the instantaneous transition of any state to a set of singular states. The algebra of operators, on which the continuous dynamics of average values of quantum states is observed in the present work, is wider than the algebra investigated in the work \cite{Volovich-Sakbaev-2014}, since it also includes the operators of shifting the argument by an arbitrary vector.

The structure of the article is as follows.
In \S\ref{s:groups}, the notions of main groups are given, and their basic properties are discussed.
In \S\ref{s:q-groups}, the limit behaviour of these groups and their mean values is considered.
In \S\ref{s:q-walks}, the quantum states and quantum channels constructed by these groups are considered.

\section{Unitary groups and their continuity}\label{s:groups}

Let \(\beta\) be a non-principal ultrafilter concentrated on the infinity of the Real axis $\R$.
Let \(\one_A\) be the \defing{indicator function} of a set \(A\subset \R\)
(we identify \(\one_{\{x\}}\) and \(\one_x\) for \(x\in\R\)).
Denote by \(\mu_\beta\) the \defing{Banach--Ces\`{a}ro measure}, defined on the Borel $\sigma$-algebra \(\mathcal{B}(\R) \subset 2^\R\) by
\[
    \mu_\beta(A):=\lim_{\beta}\left( \frac{1}{2x} \int_{-x}^x\one_A(t)dt\right).
\]

Denote by \(\mathcal{R}_\mathcal{B}\) the completion of the $\sigma$-algebra \(\mathcal{B}(\R)\) by the measure \(\mu_\beta\).
Let \(\Hcal_\beta := \Lcal_2(\R, \mathcal{R}_\mathcal{B}, \mu_\beta, \Cx)\) be a Hilbert space of square-integrable over $\mu_\beta$ functions \(\R\to\Cx\).
It is known \cite{Glazatov-Sakbaev-2022} that it can be defined as the completion of simple \(\mathcal{B}(\R)\)-measurable maps \(\R\to\Cx\) by the Hilbert norm
\[
    \one_A \mapsto \sqrt{\mu_\beta(A)}.
\]
It is clear \cite{EDzhenzherSakbaev26} that \(\Lcal_\infty(\R) \subset \Hcal_\beta\).

In \cite{EDzhenzherSakbaev26}, it was shown that a system
\begin{equation}\label{eq:sys}
    \{ f_p\colon x\mapsto e^{ipx}\mid p\in {\mathbb R}\}
\end{equation}
forms in the space \(\Hcal_\beta\) the continuum orthonormal system, which is not complete in \(\Hcal_\beta\).
However, this system forms the orthonormal basis in the closed subspace
\[
    \Hcal^t \subset \Hcal_\beta
\]
independent of the ultrafilter \(\beta\).

Denote by \(\nu\colon 2^\R\to\N\cup\{+\infty\}\) the \defing{counting measure} defined by
\[
    \nu(A) := \abs{A};
\]
strictly speaking, we mean that
\[
    \nu(A) = \begin{cases}
        \abs{A}, &\text{if $A\subset \R$ is finite, and} \\
        +\infty, &\text{otherwise}.
    \end{cases}
\]
Define the Fourier transform \(\Fcal\colon \Hcal_\beta\to\hat\Hcal\) by
\[
    \Fcal u(y) := \int u(x)e^{-ixy}\,\mu_\beta(dx),
\]
and the inverse Fourier transform \(\Fcal^{-1}\colon \hat\Hcal\to\Hcal_\beta\) by
\[
    \Fcal^{-1}\hat{u}(x) := \int \hat{u}(y)e^{ixy}\,\nu(dy).
\]
The Fourier transform unitary maps the subspace \(\Hcal^t\) onto \(\hat\Hcal\), where \(\hat\Hcal\) is the space of square-integrable over \(\nu\) functions \(\R\to\Cx\);
in other words, a function \(\hat u\) lies in \(\hat\Hcal\) if it can be represented as
\[
    \hat u = \sum_{n=1}^\infty c_n \one_{p_n},
\]
where \(\{c_n\}\in\ell_2(\N,\Cx)\) and \(\{p_n\}\colon \N\to\R\).
Moreover,
\[
    (\Hcal^t)^\perp = \Ker\Fcal,
\]
which is clear from the definitions of \(\Hcal^t\) and \(\Fcal\).

In the space \(\Hcal_\beta\) the \defing{shift operator} \(\Sbf_h\) along any \(h\in\R\) defined by
\[
    \Sbf_h u(x) = u(x+h)
\]
is unitary.
Hence, the family \(\{\Sbf_{th}\}_{t\in\R}\) forms a one-parametrized unitary group, which is strongly continuous in \(\Hcal_\beta\) if and only if \(h=0\) (see~\cite{EDzhenzherSakbaev26}).

\begin{proposition}
    The unitary group \(\{\Sbf_{th}\}_{t\in\R}\) is strongly contunious in \(\Hcal^t\).
\end{proposition}

\begin{proof}
    The elements of the continuum system~\eqref{eq:sys} form the system of eigenvectors of an operator \(\Sbf_{h}\).
    Precisely,
    \[
        \Sbf_h f_p = e^{iph}f_p
        \quad\text{for any}\quad p\in\R.
    \]
    Hence
    \[
        \norm{\Sbf_h f_p - f_p} = \abs{e^{iph}-1}\norm{f_p} \xrightarrow[h\to0]{} 0.
    \]
\end{proof}

The Fourier transform implements the unitary equivalence of the group \(\{\Sbf_{th}\}_{t\in\R}\) acting on \(\Hcal^t\), and the group \(\{\hat\Mbf_{th}\}_{t\in\R}\), acting on functions
\[
    \hat f_p \in \hat\Hcal,\quad p\in\R,
\]
by the rule
\[
    \hat\Mbf_{h} \hat f_p = e^{iph}\hat f_p.
\]
Strictly speaking, the unitary operator \(\hat\Mbf_h\) for \(h\in\R\) is defined on \(\hat\Hcal\) by
\[
    \hat\Mbf_h\hat u(x) = e^{ihx}\hat{u}(x).
\]

Let \(a\in\R\).
Now consider on the space \(\Hcal_\beta\) the unitary operator \(\Mbf_a\) of multiplication on the function \(x\mapsto e^{iax}\).
Such an operator is well-defined since it is defined on the dense invariant subspace \(\Lcal_\infty(\R)\) and is unitary, and hence can be extended to the whole \(\Hcal_\beta\).
The family \(\{\Mbf_{ta}\}_{t\in\R}\) forms an unitary group.

\begin{proposition}
    The group \(\{\Mbf_{ta}\}_{t\in\R}\) is strongly continuous if and only if $a=0$.
\end{proposition}

\begin{proof}
    For simplicity, let \(a=1\).
    It is sufficient to consider some \(f_p\in\Hcal^t\). Then
    \[
        \norm{\Mbf_{t+s} f_p - \Mbf_t f_p} = \int \abs{e^{isx} - 1}^2\,\mu_\beta(dx) = \lim_\beta \frac{1}{x}\left(x-\frac{\sin(sx)}{s}\right)=1.
    \]
\end{proof}

This discontinuous unitary group is unitary equivalent (by the Fourier transform) to the unitary group of the operators
\[
    \hat\Sbf_{ta} = \Fcal \Mbf_{ta}\Fcal^{-1},
\]
acting on the functions \(\hat u\in\hat\Hcal\) by the rule
\[
    \hat\Sbf_{a} \hat u(x) = \hat u(x+a).
\]

Consider now the operator \(\Hbf_h\), acting on the space \(\Hcal^t\) by the rule
\[
    \Hbf_h f_p = ihpf_p.
\]
This operator generates the group \(\Sbf_{th}\) of shifts.
Hence,
\[
    \frac{d}{dt} \hat\Mbf_{th} = ihp\hat f_p.
\]
In the invariant supspace
\[
    \Hcal_\beta \ominus \Hcal^t
\]
the group \(\{\Sbf_{th}\}_{t\in\R}\) is not strongly continuous, so we cannot say about the generator of the group on this supspace.
All the same is with the group \(\Mbf_{ta}\) on \(\Hcal_\beta\).

Note that
\[
    \hat\Sbf_h \hat\Mbf_a = e^{iah}\hat\Mbf_a\hat\Sbf_h,
\]
so their groups realize the canonical commutative Weyl relations.

\section{Quantum semigroups}\label{s:q-groups}

Now fix a probability space \((\Omega, \sigma_\Omega, \Pr)\).

Consider quantum random walks of the type
\[
    \Ubf_n(t)=\Sbf_{{\sqrt\frac{t}{n}}\xi_n}\circ \ldots\circ \Sbf_{{\sqrt\frac{t}{n}}\xi_1},\quad t\geq 0,\ n\in \N,
\] 
where \(\xi,\xi_1,\ldots,\xi_n\) are independent identically distributed (i.i.d.) $\R$-valued random variables.
Denote
\[
    \hat\Ubf_n(t) := \Fcal \Ubf_n(t)\Fcal^{-1}.
\]
Then
\[
    \hat\Ubf_n(t) = \hat\Mbf_{{\sqrt\frac{t}{n}}\xi_n}\circ \ldots\circ \hat\Mbf_{{\sqrt\frac{t}{n}}\xi_1}.
\]

For \(\mathbf{K}\in\{\Mbf, \hat\Mbf, \Sbf, \hat\Sbf\}\) define the average of the random operator-valued function \(\mathbf{K}_{\sqrt{t}\xi}\) as
\[
    \Fbf_\mathbf{K}(t) := \Expect \mathbf{K}_{\sqrt{t}\xi}.
\]

Denote by \(\chi_\xi\colon\R\to\Cx\) the characteristic function of a random variable \(\xi\), defined by
\[
    \chi_\xi(x) = \Expect e^{ix\xi}.
\]

\begin{proposition}\label{p:chi-mult}
    We have
    \[
        \Fbf_{\hat\Mbf}(t)\hat u(x) = \chi_\xi(\sqrt{t}x) \hat{u}(x)
    \]
    and
    \[
        \Fbf_{\Mbf}(t) u(x) = \chi_\xi(\sqrt{t}x) {u}(x).
    \]
\end{proposition}

\begin{proof}
    Since
    \[
        \hat\Mbf_{\sqrt{t}h}\hat u(x)=e^{i\sqrt t hx}\hat u(x)
    \]
    for any function \(\hat u\in \hat\Hcal\),
    \[
        \Fbf_{\hat\Mbf}(t)\hat u(x)= \Expect e^{i\sqrt{t}\xi x}\hat{u}(x) = \chi_\xi(\sqrt{t}x) \hat{u}(x).
    \]
    The second claim is proved analogously.
\end{proof}



\begin{lemma}\label{l:zero-shift}
    Suppose that \(\xi\) does not have a discrete part.
    Then the family \(\left\{\Fbf_{\hat\Sbf}(t)\right\}_{t\geq0}\) on \(\hat\Hcal\) forms a semigroup
    \[
        \Fbf_{\hat\Sbf}(t) = \begin{cases}
            \mathbf{I}, &t=0,\\
            \mathbf{0}, &t>0.
        \end{cases}
    \]
\end{lemma}

\begin{proof}
    It is sufficient to prove that
    for any $t>0,  \ \hat u,\hat v\in \hat \Hcal$ we have
    \begin{equation}\label{eq:zero}
        (\Fbf_{\hat\Sbf}(t)\hat u, \hat v) = 0.
    \end{equation}
    Since \(\hat u,\hat v\in\hat\Hcal\), there exist some countable sets \(M_u,M_v\subset\R\) such that \(\hat u=0\) outside \(M_u\), and analogously for \(\hat v\).
    Then the event
    \[
        (\hat\Sbf_{\sqrt{t}\xi}\hat u, \hat v)\neq 0
    \]
    is countable.
\end{proof}

\begin{lemma}
    Suppose that \(\xi\) does not have a discrete part.
    Then the family \(\left\{\Fbf_\Mbf(t)\right\}_{t\geq0}\) on \(\Hcal^t\) forms a semigroup
    \[
        \Fbf_\Mbf(t) = \begin{cases}
            \mathbf{I}, &t=0,\\
            \mathbf{0}, &t>0.
        \end{cases}
    \]
\end{lemma}

\begin{proof}
    Note for \(\xi\sim N(0,D)\) we have \(\chi_\xi(x) = e^{-\frac{1}{2}x^2D}\), and so
    \[
        \Fbf_\Mbf(t)f_p(x) = e^{-\frac{1}{2}tx^2D}f_p(x)=0,
    \]
    where the latter equality follows since functions \(e^{-x^2}\) and \(e^{ipx}\) are orthogonal.
    For the general \(\xi\) use the previous lemma and the unitary equivalence of random walks of the kind \(\hat\Sbf_h\) in \(\hat\Hcal\) and random exponents of the kind \(\Mbf_h\) in \(\Hcal^t\).
\end{proof}

It is interesting to investigate these semigroups in the spirit of the Chernoff Product Formula.

\begin{lemma}
    Suppose that \(\Expect\xi = 0\) and \(0 < \Expect\xi^2 =: D < +\infty\).
    Then the sequence of \(\Expect\hat\Ubf_n(t)\) converges in SOT of \(B(\hat\Hcal)\) uniformly on any segment of \(\R_+\) to the semigroup of the operators of multiplication on the function
    \[
        x\mapsto e^{-\frac{1}{2}tDx^2}.
    \]
\end{lemma}

\begin{proof}
    The pointwise convergence follows since
    \[
        \Expect\hat\Ubf_n(t) = \left( \Fbf_{\hat\Mbf}\left(\frac{t}{n}\right) \right) ^n,
    \]
    where the last is the operator of the multiplication on the function
    \[
        x \mapsto \chi_\xi\left(\frac{\sqrt{t}x}{\sqrt{n}}\right)^n = \left(1-\frac{tDx^2}{2n} + o(1/n)\right)^n \xrightarrow[n\to\infty]{} e^{-\frac{1}{2}tDx^2}.
    \]
    
    For the uniform convergence we need to apply the Chernoff Product Formula.
    It is clear that \(\Fbf_{\hat\Mbf}(0)\) is the identity operator, and that \(\norm{\Fbf_{\hat\Mbf}}_{B(\hat\Hcal)}\leq 1\) for all \(t\geq 0\).
    Moreover, since \(\Fbf_{\hat\Mbf}(t)\) acts like a multiplication by the characteristic function \(\chi_\xi\) (by Proposition~\ref{p:chi-mult}), 
    for any \(\one_p \in \hat\Hcal\) we have
    \[
        \norm{\Fbf_{\hat\Mbf}(t+s)\one_p - \Fbf_{\hat\Mbf}(t)\one_p}_{\hat\Hcal} = \abs{\chi_\xi(\sqrt{t+s}p) - \chi_\xi(\sqrt{t}p)}.
    \]
    This means that \(\Fbf_{\hat\Mbf}(t)\) is uniformly countinuous on the functions \(\one_p\).
    Since these functions form the orthonormal basic in \(\hat\Hcal\), we obtain that \(\Fbf_{\hat\Mbf}(t)\) is countinuous on the functions \(\hat u\in\hat\Hcal\), and uniformly continuous on functions \(\hat u \in \hat\Hcal_{fin}\), where \(\hat\Hcal_{fin} \subset \hat\Hcal\) is the dense subspace of functions of the kind
    \[
        \sum_{i=1}^n c_i \one_{p_i} \in \hat\Hcal_{fin}.
    \]
    Finally, denoting by $\hat{\mathbf{L}}\colon \hat\Hcal_{fin}\to\hat\Hcal_{fin}$ the operator of multiplication on the function
    \[
        x\mapsto -\frac{1}{2}Dx^2,
    \]
    for \(\hat u\in\hat\Hcal_{fin}\) we obtain that the limit
    \[
        \lim_{t\to+0} \frac{\Fbf_{\hat\Mbf}(t)\hat u - \hat u}{t} = \hat{\mathbf{L}} \hat u
    \]
    exists.
    Hence, by the Chernoff Product Formula \cite{Chernoff-1968} (see also \cite[Theorem~5.2]{Engel-Nagel-sec3}), we have that \(\hat{\mathbf{L}}\) is closable, and its closure generates the bounded strongly continuous semigroup \(e^{t\hat{\mathbf{L}}}\), which is given by
    \[
        e^{t\hat{\mathbf{L}}}\hat u = \lim_{n\to\infty} \Fbf_{\hat\Mbf}\left(\frac{t}{n}\right)^n \hat u = \lim_{n\to\infty} \Expect\hat\Ubf_n(t)\hat u
    \]
    for \(\hat u\in\hat\Hcal\) and uniformly for \(t\) in any segment.
\end{proof}

\begin{lemma}
    Suppose that \(\Expect\xi = 0\) and \(0 < \Expect\xi^2 =: D < +\infty\).
    Then the sequence of \(\Expect\Ubf_n(t)\) converges in SOT of \(B(\Hcal^t)\) uniformly on any segment of \(\R_+\) to the semigroup of the operators of multiplication on the function
    \[
        \Fcal^{-1}(x\mapsto e^{-\frac{1}{2}tDx^2})\Fcal.
    \]
\end{lemma}

\begin{proof}
    This again follows due to the unitary equivalence of \(\Sbf_h\) and \(\hat\Mbf_h\) via the Fourier transform.
\end{proof}



\section{Quantum channels walks}\label{s:q-walks}

In order to understand where the quantum state goes to in the process of averaging of unitary transforms \(\Ubf_n(t)\), we consider the random evolution of quantum channels.
First we give some general theory on quantum states.
For additional overview see, for example, \cite{Volovich-Sakbaev-2018}.

Consider the Banach algebra \(C(\hat\Hcal)\subset B(\hat\Hcal)\) of bounded linear operators on \(\hat\Hcal\), generated by
\begin{itemize}
    \item operators of multiplication on bounded measurable functions \(\R\to\Cx\), and 
    \item operators \(\mathbf{A}\), obtained as Pettis integral over some bounded countably additive measure \(m\colon 2^\R\to\Cx\) by
    \[
        \mathbf{A} = \int \hat\Sbf_a \,m(da).
    \]
\end{itemize}
In particular, this Banach algebra \(C(\hat\Hcal)\) contains the shift and the multiplication operators from the canonical commutative Weyl relations;
see in Example~\ref{ex:bad} the explanation of the choice of such Banach algebra.

Denote by \(\Sigma(\hat\Hcal) := S^{1+}(C(\hat\Hcal)^*)\) the space of \defing{quantum states}, that is, the linear functionals \(C(\hat\Hcal) \to \Cx\) lying on the intersection of the positive cone and the unit sphere.
We denote the \defing{action} of \(\rho \in \Sigma(\hat\Hcal)\) on \(\mathbf{A}\in C(\hat\Hcal)\) by
\[
    \braket{\rho, \mathbf{A}}.
\]

In quantum mechanics naturally arise the subspace \(\Sigma_n(\hat\Hcal) \subset S^{1+}(\Tcal_1(\hat\Hcal))\) of \defing{normal quantum states}, in other words, the nucleus operators with trace. The equivalent definition of normal states is that they are $*$-weakly continuous quantum states.
We say that a state \(\rho=\rho_{\hat u}\) is a \defing{vector pure state} if for any \(\mathbf{A}\in C(\hat\Hcal)\)
\[
    \Braket{\rho_{\hat u}, \mathbf{A}} = (\hat u, \mathbf{A}\hat u).
\]
So, this is a normal pure state.

For any \(\hat{v}\in\hat\Hcal\) denote by \(\mathbf{P}_{\hat v}\) the \defing{projector} on the pure state \(\hat{v}\).
It is the multiplication operator on the bounded measurable function, and can be described by
\[
    (\hat{u}, \mathbf{P}_{\hat v}\hat{u}) = \abs{(\hat u, \hat v)}^2.
\]
It is well known that a state \(\rho\) is normal if
\[
    \sup_{\mathbf{P}_{fin}} \Braket{\rho, \mathbf{P}_{fin}} =1,
\]
where the supremum is taken over projectors on finitely-dimensional subspaces.
If this equality does not hold, then \(\rho\) has a singular part.
So, we say that a state \(\rho\) lies in the space \(\Sigma_s(\hat\Hcal)\) of \defing{singular quantum states} if
\[
    \Braket{\rho, \mathbf{P}_{\hat v}} =0
    \quad\text{for any \(\hat{v}\in\hat\Hcal\)}.
\]
The Yosida--Hewitt decomposition \cite{Yosida-Hewitt} claims that
\[
    \Sigma(\hat\Hcal) = \Sigma_n(\hat\Hcal) \oplus \Sigma_s(\hat\Hcal),
\]
meaning that for any \(\rho\in\Sigma(\hat\Hcal)\) there is a convex combination
\[
    \rho = p\rho_n + (1-p)\rho_s
\]
for some \(p\in [\,0,1\,]\), \(\rho_n\in\Sigma_n(\hat\Hcal)\), and \(\rho_s\in\Sigma_s(\hat\Hcal)\).

For \(h\in\R\), we define the \defing{quantum channel} \(\Tbf_h\colon \Sigma(\hat\Hcal)\to\Sigma(\hat\Hcal)\) as in \cite{AmosovSakbaev} by the action
\[
    \Braket{\Tbf_h\rho, \mathbf{A}} := \Braket{\rho, \hat\Sbf_h^*\mathbf{A}\hat\Sbf_h}.
\]
We will write shortly
\[
    \Tbf_h\rho = \hat\Sbf_h\rho\hat\Sbf_h^*.
\]
Such notation is consistent with the definition of quantum channels via Krauss operators in quantum mechanics.

\begin{proposition}\label{p:channel-act}
    Let \(\hat u\in\hat\Hcal\) and \(\rho=\rho_{\hat u}\) be a pure state.
    Let \(h\in\R\) and \(\mathbf{A}\in B(\hat\Hcal)\). Then
    \[
        \Braket{\Tbf_h\rho,\mathbf{A}} = (\hat\Sbf_h\hat{u}, \mathbf{A}\hat\Sbf_h\hat{u}).
    \]
\end{proposition}

\begin{proof}
    This holds by definition of a pure state, since the state
    \(
        \Tbf_h\rho = \hat\Sbf_h\rho_{\hat u}\hat\Sbf_h^* = \rho_{\hat\Sbf_h\hat{u}}
    \)
    is pure.
\end{proof}

Let \(\xi\colon\Omega\to\R\) be a random variable as in the previous paragraph.
Let \(\rho\in\Sigma(\hat\Hcal)\) be a pure vector state.
As it was shown, for $\xi$ without a discrete part we have \(\Expect \hat\Sbf_{\xi} \hat u=0\).
Nevertheless, the averaging of quantum states \(\Tbf_\xi\rho\) is a quantum state, that is, a linear continuous functional \(\Expect \Tbf_\xi\rho\) on the Banach algebra \(C(\hat\Hcal)\), which is non-negative on the cone \(C^+(\hat\Hcal)\) of non-negative operators, and normed by the condition
\[
    \braket{\Expect\Tbf_\xi\rho, \mathbf{I}} = 1.
\]
This average is given by the Pettis integral; that is, for \(\mathbf{A}\in C(\hat\Hcal)\) is defined by
\[
    \Braket{\Expect\Tbf_\xi\rho, \mathbf{A}} :=
    \Expect\Braket{\Tbf_\xi\rho, \mathbf{A}}.
\]

In the following lemmas~\ref{l:mult-good} and~\ref{l:shift-good}, and Corollary~\ref{cor:shift-good}, we prove the correctness of the definition.
Note that it is not obvious due to Example~\ref{ex:bad} below.

\begin{example}\label{ex:bad}
    Suppose that we defined quantum channels on the whole \(B(\hat\Hcal)\).
    Take \(\mathbf{A}\) being the multiplication on the indicator of non-measureble subset \(V\subset \R\).
    Then for \(\rho=\rho_{\one_p}\)
    \[
        \Braket{\Tbf_\xi\rho, \mathbf{A}} = \one_V(\xi),
    \]
    and thus the Pettis integral \(\Expect\Tbf_\xi\rho\) is undefined.
    That is why we consider the Banach algebra \(C(\hat\Hcal)\).
\end{example}


\begin{lemma}\label{l:mult-good}
    Let \(\hat u\in\hat\Hcal\) be of the form
    \[
        \hat u=\sum\limits_{k=1}^{\infty}c_k\one_{p_k},
    \]
    and \(\rho=\rho_{\hat u}\) be a pure state.
    Denote by \(\hat\Mbf_f\) the multiplication operator on a bounded measurable function \(f\in\Lcal_\infty(\R,\lambda)\).
    Then
    \[
        \Braket{\Expect{\Tbf}_\xi\rho,\hat\Mbf_f} = \Braket{\rho, \hat\Mbf_g},
    \]
    where
    \[
        g(x) = \Expect f(\xi-x).
    \]
\end{lemma}

\begin{proof}
    First consider the case  \(\hat u=\one_0\).
    The idea of the proof is classical: to consider \(f\) being an indicator, then a simple function, and finally arbitrary function.
    So, let $A\in\mathcal{B}(\R)$ be a Borel subset of $\R$, and put \(f = \one_A\).
    Then \(\hat\Mbf_f\) is the orthogonal projector in \(\hat\Hcal\), defined as a multiplication operator on the indicator function \(\one_A\).
    Hence we have
    \[
        \Braket{\Expect{\Tbf}_\xi\rho,\hat\Mbf_f} =
        \Expect (\hat\Sbf_{\xi}\hat u, \hat\Mbf_f\hat\Sbf_{\xi}\hat u)=
        \Expect \one_A(\xi) =
        \Pr\{\xi\in A\},
    \]
    where
    \begin{itemize}
        \item the first equality is by definition of the Pettis integral, and by Proposition~\ref{p:channel-act},
        \item the second equality holds since \((\hat\Sbf_h \one_0, \hat\Mbf_f\hat\Sbf_h\one_0) = (\one_h, \one_A\one_h) = \one_A(h)\), and proves the measurability of \(\Braket{{\Tbf}_\xi\rho,\hat\Mbf_f}\), and
        \item the third equality is by definition of the expected value.
    \end{itemize}
    Thus, by linearity, for a simple function $f$ we have
    \[
        \Braket{\Expect{\Tbf}_\xi\rho,\hat\Mbf_f} =
        \Expect (\hat\Sbf_{\xi}\hat u, \hat\Mbf_f\hat\Sbf_{\xi}\hat u)=
        \Expect f(\xi).
    \]
    Hence, by taking the limit, we have the analogous equality for any measurable bounded function $f$.

    Second, if \(u=\one_{p_0}\) for some \(p_0\in\R\), then
    \[
        \Braket{\Expect{\Tbf}_\xi\rho,\hat\Mbf_f} =
        \Expect f(\xi-p_0).
    \]
    
    Finally, consider arbitrary
    \[
        \hat u=\sum\limits_{k=1}^{\infty}c_k\one_{p_k},
    \]
    where \(\{p_k\}\colon\N\to\R\), and \(\{c_k\}\in \ell_2(\N,\Cx)\).
    Then
    \[
        (\one_{p_j}, \hat\Mbf_f\one_{p_k}) = f(p_k)(\one_{p_j}, \one_{p_k}) = f(p_k)\one_{p_j}(p_k),
    \]
    which implies
    \[
        \Braket{\rho, \hat\Mbf_f} = \sum_{k=1}^\infty \abs{c_k}^2 f(p_k).
    \]
    Thus
    \[
        \Braket{\Expect{\Tbf}_\xi\rho,\hat\Mbf_f} = \sum_{k=1}^\infty \abs{c_k}^2 \Expect f(\xi-p_k) = \Braket{\rho, \hat\Mbf_g}.
    \]
\end{proof}

\begin{lemma}\label{l:shift-good}
    Let \(\hat u\in\Hcal\) and \(\rho=\rho_{\hat u}\) be a pure state.
    Then the equality
    \[
        \Braket{\Expect\Tbf_\xi\rho, \hat\Sbf_a} = \Braket{\rho, \hat\Sbf_a}
    \]
    holds.
    Moreover, these terms are non-zero for $a\in R$ from some countable subset \(R\subset \R\), and are independent of $t$.
\end{lemma}

\begin{proof}
    Clearly,
    \[
        (\one_{x}, \one_y) = \delta_{x,y},
    \]
    where \(\delta\) is the Kronecker delta.
    Then for any \(h\in\R\)
    \[
        (\hat\Sbf_h\one_{x}, \hat\Sbf_a\hat\Sbf_h\one_{y}) = (\one_{x+h}, \one_{y+h+a}) = \delta_{x+h,y+h+a} = \delta_{x,y+a}.
    \]
    In particular,
    \[
        (\hat\Sbf_h\one_{x}, \hat\Sbf_a\hat\Sbf_h\one_{y}) = (\one_{x}, \hat\Sbf_a\one_{y}).
    \]
    This means the independence of our terms of shifts.
    Then for
    \[
        \hat u = \sum_{j=1}^\infty c_j\one_{p_j}
    \]
    we have
    \[
        (\hat\Sbf_h\hat u, \hat\Sbf_a\hat\Sbf_h\hat u) =
        \sum_{j,k} c_j\overline{c_k} (\hat\Sbf_h\one_{p_j}, \hat\Sbf_a\hat\Sbf_h\one_{p_k}) =
        \sum_{j,k} c_j\overline{c_k} (\one_{p_j}, \hat\Sbf_a\one_{p_k}) =
        (\hat u, \hat\Sbf_a\hat u).
    \]
    Thus,
    \[
        \Braket{\Expect{\Tbf}_\xi\rho,\hat\Sbf_a} =
        \Expect (\hat\Sbf_{\sqrt{t}\xi}\hat u, \hat\Sbf_a\hat\Sbf_{\sqrt{t}\xi}\hat u)=
        \Expect (\hat u, \hat\Sbf_a\hat u) =
        (\hat u, \hat\Sbf_a\hat u) = \Braket{\rho, \hat\Sbf_a},
    \]
    where
    \begin{itemize}
        \item the first equality is by definition of the Pettis integral, and by Proposition~\ref{p:channel-act},
        \item the second equality follows, since we just proved the insignificance of the shift, and proves the measurability of \(\Braket{{\Tbf}_\xi\rho,\hat\Sbf_a}\), and
        \item the third equality holds since the expected value of a constant is the constant,
        \item the last equality is by definition of a pure vector state.
    \end{itemize}
\end{proof}

\begin{corollary}\label{cor:shift-good}
    Let \(\hat u\in\hat\Hcal\) and \(\rho=\rho_{\hat u}\).
    Let an operator \(\mathbf{A} = \int \hat \Sbf_a\,m(da)\) be a convolution of \(\hat\Sbf_a\) with some other countably additive measure $m$.
    Then
    \[
        \Braket{\Expect{\Tbf}_\xi\rho,\mathbf{A}} = \Braket{\rho, \mathbf{A}},
    \]
    and the last expression is non-zero only for measures $m$ with atoms.
\end{corollary}

\begin{proof}
    By definition of \(\mathbf{A}\),
    \[
        \Braket{\Expect{\Tbf}_\xi\rho,\mathbf{A}} = \int \Braket{\Expect{\Tbf}_\xi\rho,\hat\Sbf_a} \,m(da).
    \]
    Now apply Lemma~\ref{l:shift-good} to the integrand expression in order to obtain
    \[
        \Braket{\Expect{\Tbf}_\xi\rho,\mathbf{A}} = \int (\hat u, \hat\Sbf_a\hat u)\,m(da) = \Braket{\rho, \mathbf{A}}.
    \]
\end{proof}

Thus, the correctness of the definition of the quantum channel \(\Expect \Tbf_\xi\) on the pure states is proved, and we are ready to define it on all quantum states.
Denote by \(S^{1+}(ba(S^1(\hat\Hcal)))\) the Banach space of finitely additive non-negative normed measures with bounded variation on the measurable space \((S^1(\hat\Hcal),2^{S^1(\hat \Hcal)})\).
It is known, see for example \cite[Theorem~2.2]{AmosovSakbaev}, that any state \(\rho\in\Sigma(\hat\Hcal)\) can be represented as the Pettis integral
\[
    \rho = \int_{S^1(\hat\Hcal)} \rho_{\hat u}\, \mu(d\hat u),
\]
where \(\mu\in S^{1+}(ba(S^1(\hat\Hcal)))\).
Moreover, if our space \(\hat\Hcal\) were separable (which it is not) and \(\mu\) were countably additive, then the state $\rho$ would have been normal; and for a normal state, the measure \(\mu\) can be chosen to be countably additive.
So, the natural way to define the quantum channel \(\Expect\Tbf_\xi\) is
\[
    \Expect\Tbf_\xi\rho :=
        \int_{S^1(\hat\Hcal)} \Expect\Tbf_\xi\rho_{\hat u}\, \mu(d\hat u),
\]
where the integral is in the Pettis sense.

\begin{lemma}
    The definition of \(\Expect\Tbf_\xi\rho\) does not depend on the choice of \(\mu\).
\end{lemma}

\begin{proof}
    It is sufficient to prove that
    \[
        \int_{S^1(\hat\Hcal)} \Braket{\Expect\Tbf_\xi\rho_{\hat u},\mathbf{A}}\, \mu_1(d\hat u) =
        \int_{S^1(\hat\Hcal)} \Braket{\Expect\Tbf_\xi\rho_{\hat u},\mathbf{A}}\, \mu_2(d\hat u)
    \]
    for \(\mathbf{A}\) being a multiplication operator, or the convolution of a shift and a measure, if
    \[
        \int_{S^1(\hat\Hcal)} \rho_{\hat u}\, \mu_1(d\hat u) = \int_{S^1(\hat\Hcal)} \rho_{\hat u}\, \mu_2(d\hat u).
    \]
    First consider \(\mathbf{A} = \int\hat\Sbf_a\,m(da)\).
    Then by Corollary~\ref{cor:shift-good}
    \begin{multline*}
        \int_{S^1(\hat\Hcal)} \Braket{\Expect\Tbf_\xi\rho_{\hat u},\mathbf{A}}\, \mu_1(d\hat u) =
        \int_{S^1(\hat\Hcal)} \Braket{\rho_{\hat u}, \mathbf{A}}\,\mu_1(d\hat u) = \\ =
        \int_{S^1(\hat\Hcal)} \Braket{\rho_{\hat u}, \mathbf{A}}\,\mu_2(d\hat u) =
        \int_{S^1(\hat\Hcal)} \Braket{\Expect\Tbf_\xi\rho_{\hat u},\mathbf{A}}\, \mu_2(d\hat u).
    \end{multline*}
    Now consider \(\mathbf{A} = \hat\Mbf_f\) and denote \(g(x) := \Expect f(\xi-x)\).
    Then by Lemma~\ref{l:mult-good}
    \begin{multline*}
        \int_{S^1(\hat\Hcal)} \Braket{\Expect\Tbf_\xi\rho_{\hat u},\mathbf{A}}\, \mu_1(d\hat u) =
        \int_{S^1(\hat\Hcal)} \Braket{\rho_{\hat u}, \hat\Mbf_g}\,\mu_1(d\hat u) = \\ =
        \int_{S^1(\hat\Hcal)} \Braket{\rho_{\hat u}, \hat\Mbf_g}\,\mu_2(d\hat u) =
        \int_{S^1(\hat\Hcal)} \Braket{\Expect\Tbf_\xi\rho_{\hat u},\mathbf{A}}\, \mu_2(d\hat u).
    \end{multline*}
\end{proof}

The following theorem shows that the quantum channel \(\Expect \Tbf_\xi\) is fabulous, since in many cases it produces singular states.

\begin{theorem}
    Suppose $\xi$ does not have a discrete part.
    Then the quantum state \(\Expect\Tbf_\xi\rho\) is singular.
\end{theorem}

\begin{proof}
    We need to show that for any \(\hat v\in \hat\Hcal\) we have
    \[
        \Braket{\Expect\Tbf_{\xi}\rho, \mathbf{P}_{\hat v}} =0.
    \]
    If \(\rho=\rho_{\hat u}\) is a vector pure state, then
    \[
        \Braket{\Expect\Tbf_{\xi}\rho, \mathbf{P}_{\hat v}} =
        \Expect (\Sbf_{\xi}\hat u, \mathbf{P}_{\hat v}\Sbf_{\xi}\hat u)=
        \Expect \abs{(\hat\Sbf_{\xi}\hat u,\hat v)}^2=0,
    \]
    where
    \begin{itemize}
        \item the first equality is by definition of the Pettis integral, and by Proposition~\ref{p:channel-act},

        \item the second equality is by definition of the projector, and

        \item the last equality can be shown analogously to Lemma~\ref{l:zero-shift}.
    \end{itemize}
    If \(\rho\) is given by
    \[
        \rho = \int_{S^1(\hat\Hcal)} \rho_{\hat u}\, \mu(d\hat u),
    \]
    then
    \[
        \Braket{\Expect\Tbf_{\xi}\rho, \mathbf{P}_{\hat v}} =
        \int_{S^1(\hat\Hcal)} \Braket{\Expect\Tbf_\xi\rho_{\hat u}, \mathbf{P}_{\hat v}}\, \mu(d\hat u) = 0,
    \]
    since the integrand term is zero.
\end{proof}

Now consider the random process \(\{\xi_t\}_{t\geq 0}\) of independent random variables with the property
\begin{itemize}
    \item \(\xi_0 = 0\), and
    \item \(\xi_s+\xi_t \stackrel{d}{=} \xi_{s+t}\).
\end{itemize}
For example, it could be independent \(\xi_t\sim N(0,t)\), or \(\xi_t\sim\mathrm{Cauchy}(t)\).
Now consider the quantum channel \(\Tbf(t)\), for any \(\rho\in\Sigma(\hat\Hcal)\) defined by
\[
    \Tbf(t)\rho := \Expect \Tbf_{\xi_t}\rho.
\]

\begin{theorem}\label{t:gaus-semigroup}
    The quantum channels \(\Tbf(t)\) form a semigroup for \(t\geq 0\).
\end{theorem}

\begin{proof}
    It is clear that \(\Tbf(0)\) is the identity operator.
    By \cite[Exercise~1.11]{Conway}, we may go from the Schr\"odinger representation to the Heisenberg representation; that is, we may write
    \[
        \Braket{\Expect\Tbf_\xi \rho, \mathbf{A}} = \Braket{\rho, \Expect \hat\Sbf_\xi\mathbf{A}\hat\Sbf_\xi^*},
    \]
    where both sides should be interpreted as Pettis integrals.
    Fix any \(t,s\geq 0\).
    Since \(\Tbf(t)\) is a quantum channel, it is sufficient to prove that
    \[
        \Tbf(t)\Tbf(s)\rho = \Tbf(t+s)\rho
    \]
    for \(\rho = \rho_{\hat u}\) being a pure vector state.
    Then
    \[
        \Braket{\Expect\Tbf(s) \rho, \mathbf{A}} = \int\Braket{\rho, \hat\Sbf_x\mathbf{A}\hat\Sbf_x^*}\,\Pr_{\xi_s}(dx),
    \]
    and hence
    \[
        \Tbf(t)\Tbf(s)\rho =
        \int \Braket{\rho, \hat\Sbf^*_{y}\hat\Sbf^*_{x}\mathbf{A}\hat\Sbf_{x}\hat\Sbf_{y}} \,\Pr_{\xi_s}(dx)\Pr_{\xi_t}(dy) =
        \int \Braket{\rho, \hat\Sbf^*_{z}\mathbf{A}\hat\Sbf_{z}} \,\Pr_{\xi_s+\xi_t}(dz) = \Tbf(t+s)\rho,
    \]
    where
    \begin{itemize}
        \item the first equality is by definition of \(\Tbf(t)\) and \(\Tbf(s)\),

        \item the second equality holds by the Fubini's theorem for \(z=x+y\), and

        \item the last equality follows since \(\xi_t+\xi_s \stackrel{d}{=} \xi_{t+s}\).
    \end{itemize}
\end{proof}

Now let us consider another quantum channel, for any \(\rho\in\Sigma(\hat\Hcal)\) given by
\[
    \Phbf_h \rho := \hat\Mbf_h\rho \hat\Mbf_h^*.
\]
We are interested in the quantum channel
\[
    \Expect \Phbf_{\xi}\rho,
\]
which is defined as the Pettis integral for all \(\mathbf{A}\in C(\hat\Hcal)\) by
\[
    \Braket{\Expect\Phbf_\xi\rho, \mathbf{A}} := \Expect \Braket{\Phbf_\xi\rho, \mathbf{A}}.
\]
The second term is well defined since the map \(h\mapsto \hat\Mbf_h\) is strongly continuous.

\begin{theorem}
    The quantum channel \(\Expect\Phbf_\xi\) preserves the normal state.
\end{theorem}

\begin{proof}
    It is sufficient to check this for a pure state.
    So, consider the case \(\hat u\in \hat\Hcal\) and \(\rho=\rho_{\hat u}\).
    Note that \(\hat\Mbf_h \rho_{\hat u}\hat\Mbf_h^* = \rho_{\hat\Mbf_h\hat u}\) and \(\hat\Mbf_h\one_a = e^{iha}\one_a\).
    Then for
    \[
        \hat u = \sum_{j=1}^\infty c_j\one_{p_j}
    \]
    we have
    \[
        \hat\Mbf_h\hat u = \sum_{j=1}^\infty c_je^{ihp_j}\one_{p_j},
    \]
    and so
    \[
        \hat\Mbf_h \rho_{\hat u}\hat\Mbf_h^* =
        \sum_j \abs{c_j}^2\one_{p_j} + \sum_{j\neq k} c_j\overline{c_k} e^{ih(p_j-p_k)}\ket{\one_{p_j}}\bra{\one_{p_k}}.
    \]
    Thus
    \[
        \Expect\Phbf_\xi\rho_{\hat u} = \sum_j \abs{c_j}^2\one_{p_j} + \sum_{j\neq k} c_j\overline{c_k} \chi_\xi(p_j-p_k)\ket{\one_{p_j}}\bra{\one_{p_k}},
    \]
    which is the normal state.
\end{proof}

Analogously to \(\Tbf(t)\), we may define the channel \(\Phbf(t)\) by
\[
    \Phbf(t)\rho := \Expect \Phbf_{\xi_t}\rho.
\]

\begin{theorem}
    The quantum channels \(\Phbf(t)\) form a semigroup for \(t\geq 0\).
\end{theorem}

\begin{proof}
    The proof is analogous to the proof of Theorem~\ref{t:gaus-semigroup}.
\end{proof}

\begin{remark}
    The channel \(\Phbf(t)\) can be defined on the whole Banach algebra \(B(\hat\Hcal)\) instead of just \(C(\hat\Hcal)\).
    It will preserve the normal state and form a semigroup.
\end{remark}

\printbibliography

\end{document}